\PassOptionsToPackage{dvipsnames,svgnames,x11names}{xcolor}
\documentclass[
  12pt]{article}

\usepackage{amscd,amsthm,amsfonts}
\usepackage{mathrsfs}
\usepackage{upgreek}
\usepackage{enumitem}
\usepackage{algorithmicx,algorithm}
\usepackage{makecell}
\usepackage{multirow}
\usepackage[utf8]{inputenc}

\usepackage{tikz}
\usetikzlibrary{positioning,chains,fit,shapes,calc}
\usepackage{setspace}

\newtheorem{Corollary}{Corollary}
\newtheorem{Proposition}{Proposition}

\usepackage{amsmath,amssymb}
\usepackage{iftex}
\ifPDFTeX
  \usepackage[T1]{fontenc}
  \usepackage[utf8]{inputenc}
  \usepackage{textcomp} 
\else 
\fi
\usepackage{lmodern}
\ifPDFTeX\else  
\fi
\IfFileExists{upquote.sty}{\usepackage{upquote}}{}
\IfFileExists{microtype.sty}{
  \usepackage[]{microtype}
  \UseMicrotypeSet[protrusion]{basicmath} 
}{}
\makeatletter
\@ifundefined{KOMAClassName}{
  \IfFileExists{parskip.sty}{%
    \usepackage{parskip}
  }{
    \setlength{\parindent}{0pt}
    \setlength{\parskip}{6pt plus 2pt minus 1pt}}
}{
  \KOMAoptions{parskip=half}}
\makeatother
\usepackage{xcolor}
\makeatletter
\ifx\paragraph\undefined\else
  \let\oldparagraph\paragraph
  \renewcommand{\paragraph}{
    \@ifstar
      \xxxParagraphStar
      \xxxParagraphNoStar
  }
  \newcommand{\xxxParagraphStar}[1]{\oldparagraph*{#1}\mbox{}}
  \newcommand{\xxxParagraphNoStar}[1]{\oldparagraph{#1}\mbox{}}
\fi
\ifx\subparagraph\undefined\else
  \let\oldsubparagraph\subparagraph
  \renewcommand{\subparagraph}{
    \@ifstar
      \xxxSubParagraphStar
      \xxxSubParagraphNoStar
  }
  \newcommand{\xxxSubParagraphStar}[1]{\oldsubparagraph*{#1}\mbox{}}
  \newcommand{\xxxSubParagraphNoStar}[1]{\oldsubparagraph{#1}\mbox{}}
\fi
\makeatother

\usepackage{longtable,booktabs,array}
\usepackage{calc} 
\usepackage{etoolbox}
\makeatletter
\patchcmd\longtable{\par}{\if@noskipsec\mbox{}\fi\par}{}{}
\makeatother
\IfFileExists{footnotehyper.sty}{\usepackage{footnotehyper}}{\usepackage{footnote}}
\makesavenoteenv{longtable}
\usepackage{graphicx}
\makeatletter
\def\maxwidth{\ifdim\Gin@nat@width>\linewidth\linewidth\else\Gin@nat@width\fi}
\def\maxheight{\ifdim\Gin@nat@height>\textheight\textheight\else\Gin@nat@height\fi}
\makeatother
\setkeys{Gin}{width=\maxwidth,height=\maxheight,keepaspectratio}
\makeatletter
\def\fps@figure{htbp}
\makeatother

\makeatletter
\@ifpackageloaded{caption}{}{\usepackage{caption}}
\AtBeginDocument{%
\ifdefined\contentsname
  \renewcommand*\contentsname{Table of contents}
\else
  \newcommand\contentsname{Table of contents}
\fi
\ifdefined\listfigurename
  \renewcommand*\listfigurename{List of Figures}
\else
  \newcommand\listfigurename{List of Figures}
\fi
\ifdefined\listtablename
  \renewcommand*\listtablename{List of Tables}
\else
  \newcommand\listtablename{List of Tables}
\fi
\ifdefined\figurename
  \renewcommand*\figurename{Figure}
\else
  \newcommand\figurename{Figure}
\fi
\ifdefined\tablename
  \renewcommand*\tablename{Table}
\else
  \newcommand\tablename{Table}
\fi
}
\@ifpackageloaded{float}{}{\usepackage{float}}
\floatstyle{ruled}
\@ifundefined{c@chapter}{\newfloat{codelisting}{h}{lop}}{\newfloat{codelisting}{h}{lop}[chapter]}
\floatname{codelisting}{Listing}

\makeatother
\makeatletter
\@ifpackageloaded{caption}{}{\usepackage{caption}}
\@ifpackageloaded{subcaption}{}{\usepackage{subcaption}}
\makeatother

\ifLuaTeX
  \usepackage{selnolig}  
\fi
\usepackage[]{natbib}
\usepackage{bookmark}

\IfFileExists{xurl.sty}{\usepackage{xurl}}{} 
\hypersetup{
  pdftitle={Title},
  pdfauthor={Author 1; Author 2},
  pdfkeywords={3 to 6 keywords, that do not appear in the title},
  colorlinks=true,
  linkcolor={blue},
  filecolor={Maroon},
  citecolor={Blue},
  urlcolor={Blue},
  pdfcreator={LaTeX via pandoc}}

\newcommand{\anon}{1}

\begin{document}

\def\spacingset#1{\renewcommand{\baselinestretch}%
{#1}\small\normalsize} \spacingset{1}


\if1\anon
{
  \title{\bf One-shot variable-ratio matching with fine balance}
  \author{Qian Meng \\
    Department of Statistics, University of Washington \\
    Zhe Chen \\
    Department of Biostatistics, Epidemiology and Informatics,\\
    University of Pennsylvania\\
    and \\
    Bo Zhang\thanks{Assistant Professor of Biostatistics, Vaccine and Infectious Disease Division, Fred Hutchinson Cancer Center. Email: {\tt bzhang3@fredhutch.org}. }\hspace{.2cm}\\
    Vaccine and Infectious Disease Division, Fred Hutchinson Cancer Center}
    \date{}
  \maketitle
} \fi

\if0\anon
{
  \bigskip
  \bigskip
  \bigskip
  \begin{center}
    {\LARGE\bf Title}
\end{center}
  \medskip
} \fi

\bigskip
\begin{abstract}
Variable-ratio matching is a flexible alternative to conventional $1$-to-$k$ matching for designing observational studies that emulate a target randomized controlled trial (RCT). To achieve fine balance---that is, matching treated and control groups to have the same marginal distribution on selected covariates---conventional approaches typically partition the data into strata based on estimated entire numbers and then perform a series of $1$-to-$k$ matches within each stratum, with $k$ determined by the stratum-specific entire number. This ``divide-and-conquer” strategy has notable limitations: (1) fine balance typically does not hold in the final pooled sample, and (2) more controls may be discarded than necessary. To address these limitations, we propose a one-shot variable-ratio matching algorithm. Our method produces designs with exact fine balance on selected covariates in the matched sample, mimicking a hypothetical RCT where units are first grouped into sets of different sizes and one unit within each set is assigned to treatment while others to control. Moreover, our method achieves comparable or superior balance across many covariates and retains more controls in the final matched design, compared to the ``divide-and-conquer" approach. We demonstrate the advantages of the proposed design over the conventional approach via simulations and using a dataset studying the effect of right heart catheterization on mortality among critically ill patients. The algorithm is implemented in the \textsf{R} package \textsf{match2C}.
\end{abstract}

\noindent%
\vfill

\newpage
\spacingset{1.8} 

\section{Introduction}
\label{sec: introduction}

\subsection{Matching methods in observational studies}
\label{subsec: intro review of matching methods}
A central challenge in estimating causal effects using observational data is the presence of confounders---observed and unmeasured variables that simultaneously influence both the treatment and the outcome. Reducing bias from these confounders is critical. Within the potential outcomes framework for causal inference \citep{neyman1923application,rubin1974estimating}, methods to reduce bias can be broadly divided into two categories depending on the focus: those focusing on the design stage \citep{rosenbaum2002observational,rosenbaum2010design} and those on the analysis stage (see, e.g., \citet{hernan2020causal}). Design-based approaches to confounding adjustment do not rely on outcome data and are less dependent on outcome modeling \citep{Rubin2008,stuart2023}. Among these approaches, matching---which emulates an idealized hypothetical randomized controlled trial (RCT) by constructing comparison groups balanced on observed covariates---is among the most widely used methods. Matched samples can be analyzed using randomization or biased randomization-based methods \citep{rosenbaum2002observational} or regression-based methods \citep{rubin1973use}.

Early implementations of statistical matching relied on nearest-neighbor-based greedy algorithms which choose one or more matched control units for each treated unit sequentially without revision and reconsideration of previous choices. In a seminal work, \citet{rosenbaum1989optimal} first showed that this heuristic technique can be arbitrarily worse than the optimal solution, and in the same article, \citet{rosenbaum1989optimal} first recast the problem of optimal $1$-to-$k$ matching as solving a minimum-cost flow problem on a bipartite network, a combinatorial optimization problem well studied in the operations research literature \citep{schrijver2003combinatorial}. The solution that minimizes the cost of a network designed specifically to correspond to the statistical matching problem is shown to also minimize the total within-pair distance. \citet{rosenbaum1989optimal} thus established a principled optimization backbone for much of the subsequent research. By adjusting the structure of the network --- including how edges are connected and the capacity and cost associated with the edges --- the same framework can also accommodate many additional design features; see, e.g., \citet{lu2011optimal, pimentel2018optimal,yu2020matching,zhang2023statistical}, among others.

Pair match, or $1$-to-$k$ match in general, can be rigid and not the most efficient, as the algorithm often discards many control units. Two strategies were developed to remedy this. First strategy is referred to as \emph{full matching} \citep{rosenbaum1991characterization,hansen2004full,hansen2006optimal}. A full match partitions all treated and control units into disjoint subclasses, each containing either one treated unit and multiple control units or one control unit and multiple treated units. \cite{rosenbaum1991characterization} showed that (1) any distance-minimizing subclassification can be refined into a full match; and (2) solving for an optimal full match can be formulated as a minimum-cost network flow problem with carefully designed capacity constraints. \citet{hansen2006optimal} further extended this framework to accommodate many useful design features, such as forcing minimum and maximum treatment-to-control ratio within each matched set, incorporating calipers like a propensity score caliper, and excluding certain units from the final matched groups. 

A second strategy is referred to as \emph{variable-ratio matching}. A variable-ratio match consists of a flexible hybrid of different types of $1$-to-$k$ matches \citep{pimentel2015variable}, and it is particularly useful when the number of control units only moderately exceeds some multiple of the number of treated units. For instance, with $n_t = 1000$ treated units and $n_c = 1800$ candidate control units, a pair match would construct $I = 1000$ matched pairs in the final matched design and discard $1800 - 1000 = 800$ control units, while a $1$-to-$2$ match is not feasible. A variable-ratio match comes very handy in such scenarios. For instance, a variable-ratio match may end up constructing $600$ matched pairs, $300$ $1$-to-$2$ matched sets, and $100$ $1$-to-$3$ matched sets, discarding only $1800 - 600 - 2 \times 300 - 3 \times 100 = 300$ control units. In this way, control units are utilized to a fuller extent, improving the statistical efficiency \citep{ming2000substantial}.
 
\subsection{Variable-ratio match; fine balance}
\label{subsec: intro vr match}

\cite{ming2001note} first proposed an efficient assignment algorithm to conduct a variable-ratio match. \citet{pimentel2015variable} proposed an approach to conducting a variable-ratio match that also forces a useful design feature called \emph{fine balance}. Fine balance ensures that the treated group and the matched comparison group have identical marginal distributions on target variables, without imposing constraints on individual matched pairs \citep{rosenbaum2007minimum}. Fine balance is particularly useful for variables with many categories, because the stochastic balancing property of the propensity score often works poorly for these variables.

In the first step of \citeauthor{pimentel2015variable}'s  \citeyearpar{pimentel2015variable} algorithm, an ``entire number", defined as the inverse odds of the propensity score, was estimated from data, and units (treated and control) are then stratified based on their estimated entire numbers. In the second step,  for units in the stratum with an entire number in $(0,2)$, an optimal pair match with fine balance was constructed; in the stratum with an entire number in $[2,3)$, an optimal $1$-to-$2$ match with fine balance was constructed. This process continues until an optimal $1$-to-$k$ match with fine balance is constructed within each entire number-defined stratum. 





\subsection{Application to a study of right heart catheterization; limitations of the current method}
\label{subsec: intro case study and limitation}
We illustrate \citeauthor{pimentel2015variable}'s   \citeyearpar{pimentel2015variable} method and discuss its limitations using a publicly available dataset studying the effect of right heart catheterization (RHC) on mortality \citep{connors1996effectiveness}. RHC is a diagnostic procedure used in the management of critically ill patients in intensive care units. It is intended to guide therapy by providing detailed  and important information about patients' conditions. However, \cite{connors1996effectiveness} found that RHC increased mortality among critically ill patients, and this seminal study raised concerns about the routine use of RHC. Below, we follow the case study in \citet{rosenbaum2012optimal} and illustrate \citeauthor{pimentel2015variable}'s \citeyearpar{pimentel2015variable} method using a subset of this dataset consisting of patients under the age of 65. The analysis dataset consists of $n_t = 1194$ patients who received RHC and $n_c = 1804$ who did not. The treated-to-control ratio of approximately $1.5$ corresponds to a scenario where a variable-ratio match could be potentially useful.

To illustrate, 
we consider finely balancing the insurance type of patients. The insurance type is a categorical variable consisting of $7$ levels: Medicaid, Medicare, Medicare $\&$ Medicaid, No insurance, Private insurance, Private insurance $\&$ Medicare. Table \ref{tab:ninslcas_dist} summarizes the distributions of the insurance type variable in the entire cohort, treated cohort, and control cohort. For each insurance type level, there are more control patients (No RHC) than treated patients (RHC), and therefore, fine balance on the insurance type variable is in principle feasible. 

\begin{table}[!ht]
\centering
\begin{tabular}{lccc}
\hline
\textbf{Panel A} & All patients below 65 & Treated (RHC) & Control (No RHC) \\
 Insurance Type& $N = 2998$ & $n_t = 1194$ & $n_c = 1804$ \\
\hline
Medicaid & 611 & 182 & 429 \\
Medicare & 274 & 107 & 167 \\
Medicare \& Medicaid & 141 & 55 & 86 \\
No insurance & 271 & 113 & 158 \\
Private & 1544 & 675 & 869 \\
Private \& Medicare & 157 & 62 & 95 \\
\hline
& Patients with entire \\
\textbf{Panel B} & number $\in [2, 3)$ & Treated & Control \\
& $N = 561$ & $n_t = 176$ & $n_c = 385$ \\
\hline
Medicaid                & 191 &  50 & 141 \\
Medicare                &  63 &  17 &  46 \\
Medicare \& Medicaid    &  26 &   8 &  18 \\
No insurance            &  62 &  24 &  38 \\
Private                 & 197 &  70 & 127 \\
Private \& Medicare     &  22 &   7 &  15 \\
\hline
& Patients with entire \\
\textbf{Panel C} & number $\in [4, \infty)$ & Treated & Control \\
& $N = 21$ & $n_t = 5$ & $n_c = 16$ \\
\hline
Medicaid                & 10 &  3 &  7 \\
Medicare                &  1 &  0 &  1 \\
Medicare \& Medicaid    &  1 &  0 &  1 \\
No insurance            &  4 &  1 &  3 \\
Private                 &  4 &  1 &  3 \\
Private \& Medicare     &  1 &  0 &  1 \\
\hline
\end{tabular}
\caption{\textbf{Panel A}: Marginal distributions of the insurance type variable in the entire cohort, treated patients, and control patients. \textbf{Panel B and C}: Marginal distributions of the insurance type variable in subcohorts defined by patients whose entire number $\in [2, 3)$ and $[4, \infty)$. With each subcohort, distributions are shown separately for all patients, treated patients, and control patients.}
\label{tab:ninslcas_dist}
\end{table}

Following \citet{pimentel2015variable}, we first calculated the entire number for each patient and partitioned the entire cohort ($N = 2998$) into five subcohorts according to the estimated entire number. Within the subcohort with entire number between $[k, k+1)$, we performed a $1$-to-$k$ match with fine balance. However, this strategy was not always feasible. For instance, for patients with entire numbers falling in the interval $[2,3)$ (see Panel B of Table \ref{tab:ninslcas_dist}), a $1$-to-$2$ match with fine balance on the insurance type was not feasible because there were no sufficient controls for the level ``No insurance " ($24$ treated versus $38$ control patients) or the level ``Private" ($70$ treated versus $127$ control patients). As another example, for patients with entire number in $[4, \infty)$, a $1$-to-$4$ match with fine balance was in general not feasible because controls were lacking in this subcohort. \citet{pimentel2015variable}, in their discussion, acknowledged this limitation and suggested that one reasonable solution is to reduce the number of controls. For instance, instead of conducting a $1$-to-$2$ match for the subcohort with entire number between $2$ and $3$, one may instead conduct a pair match after examining Table \ref{tab:ninslcas_dist} Panel B. 

Moreover, even if the researcher achieves fine balance within each individual $1$-to-$k$ match (possibly by selecting a smaller $k$ within each entire-number–defined stratum), there is no guarantee that the final pooled matched sample also achieves fine balance. One sufficient condition to achieve fine balance in the pooled sample requires the marginal distribution of the categorical variable in the treated group to be identical across all entire-number–defined strata---a condition that rarely holds, because the entire number is a function of the propensity score, which is typically correlated with the variable targeted for fine balance. 

Motivated by these practical challenges, we propose an efficient, one-shot variable-ratio match algorithm that simultaneously addresses all these aforementioned limitations. In particular, the proposed algorithm can achieve fine balance on the desired categorical variable (or a combination of multiple categorical variables), as long as fine balance is feasible in the entire cohort before matching. For instance, in the RHC dataset, fine balance on insurance type is feasible when the entire cohort is viewed holistically, although it is no longer feasible within each stratum defined by the entire number. Compared to \citeauthor{pimentel2015variable}'s \citeyearpar{pimentel2015variable} two-step, ``divide-and-conquer" approach, our proposed method only solves one global optimization problem, instead of dividing the problem into several sub-problems, and the global optimization problem can be solved equally efficiently as multiple smaller problems. As we will demonstrate via simulation studies and in the RHC dataset, the new algorithm helps researchers avoid a lot of \emph{ad hoc} decisions, achieves desired fine balance, and often discards fewer control units (and hence maintaining a larger matched sample) compared to the current two-step approach. 

The rest of the article is organized as follows. Section \ref{sec: methods} develops the notation, defines the network structure underlying the optimization problem, and proves that the solution to the minimum-cost network flow problem yields an optimal variable-ratio matched sample subject to the fine balance constraint. Section \ref{sec: simulation} compares the performance of the proposed approach to the existing one. The new algorithm is applied to the RHC dataset in Section \ref{sec: case study}. We conclude with a discussion in Section \ref{sec: discussion}. \textsf{R} package and code to reproduce results in the paper can be found in the Supplemental Materials.

\section{Methods}
\label{sec: methods}
\subsection{Notation; a hypothetical RCT}
\label{subsec: method notation}
We consider a setting with $T$ treated units and $C \geq T$ control units to be matched. We let $\mathcal{T} = \{\tau_1, \dots, \tau_T\}$ and $\mathcal{C} = \{\gamma_1, \dots, \gamma _C\}$ denote the treated and control units, respectively. Furthermore, suppose we have a nominal covariate $X_{\text{nom}}$ with $B$ levels that researchers target for fine balance. For each level $b \in \{1, 2, \dots, B\}$, the treated group $\mathcal{T}$ has $n_b$ units, and the control group $\mathcal{C}$ has $N_b$ units. It follows that $\sum_{b=1}^B n_b = T$ and $\sum_{b=1}^B N_b = C$. We further assume $N_b \geq n_b$ for each level $b$, and we let $\kappa \in [1, \kappa_{\max}]$, where $\kappa_{\max}$ denotes the minimum ratio of $N_b$ to $n_b$ across all $B$ levels, that is: 
\[
\kappa_{\max} = \min_{b=1,\dots,B} \frac{N_b}{n_b}.
\]

Our goal is to embed observational data into the following hypothetical randomized controlled trial. In the first step, units are grouped into sets of varying sizes. In the second step, exactly one unit from each set is randomly assigned to treatment, while the remaining units in the set are assigned to control. Under this randomization scheme, the marginal distribution of covariates, including the nominal variable $X_{\text{nom}}$, would be balanced between treated and control units.  

To emulate this hypothetical RCT, we aim to match $\sum_{b=1}^B \lfloor \kappa n_b \rfloor$ control units to $T$ treated units, where $\lfloor \cdot \rfloor$ denotes the floor function, $\kappa \in [1, \kappa_{\max}]$, and each treated unit may be matched to $1$ or more controls. To enforce fine balance on the nominal covariate $X_{\text{nom}}$, the algorithm discards 
\begin{equation}\label{eq: Mb}
    M_b = N_b - \lfloor \kappa n_b \rfloor
\end{equation}
controls for each level $b \in \{1, \dots, B\}$. In total, $\sum_{b=1}^B M_b$ controls remain unmatched. If $\kappa n_b$ is an integer for all $b$, fine balance is achieved exactly: the distribution of $X_{\text{nom}}$ is identical across treated and matched control groups. If $\kappa n_b$ is non-integer for some $b$, the imbalance persists but is negligible, arising only from rounding.

\subsection{Network structure}
\label{subsec: method network}

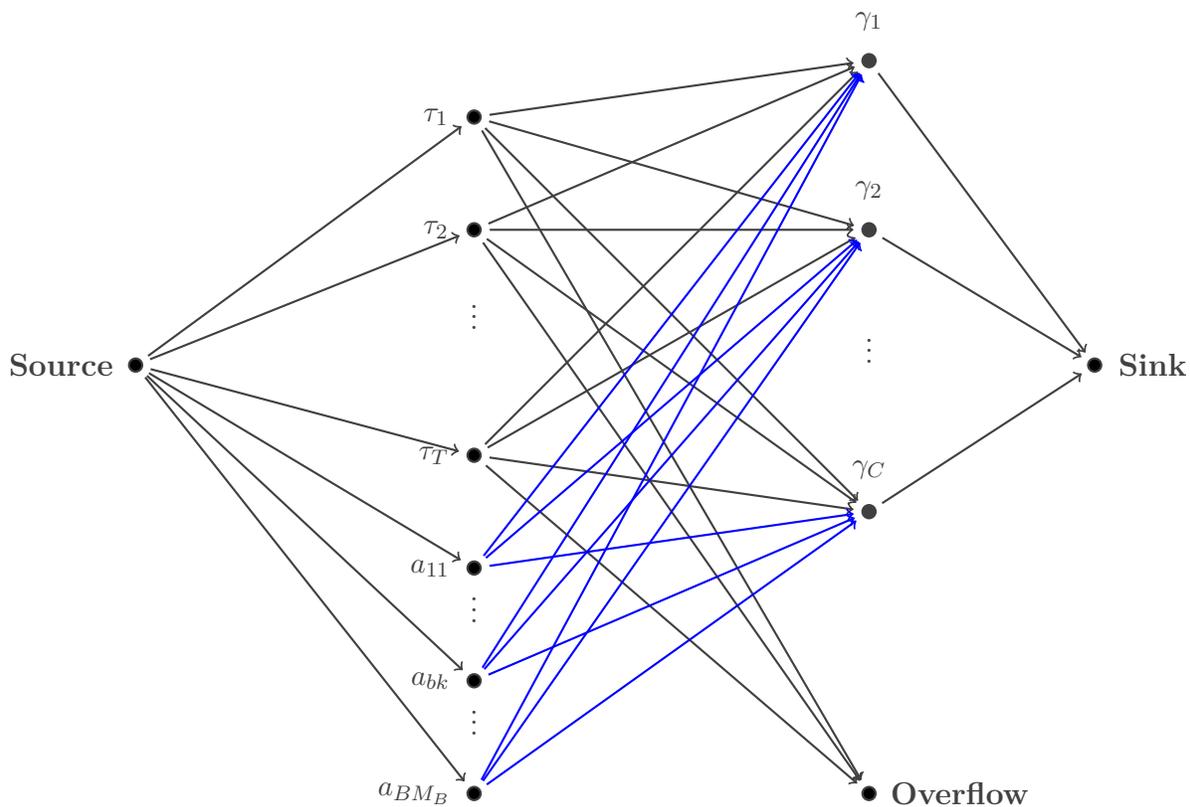
\begin{figure}[ht]
\centering

\begin{tikzpicture}[scale = 1.5, thick, color = darkgray,
  every node/.style={draw,circle},
  fsnode/.style={fill=black, inner sep = 0pt, minimum size = 5pt},
  ssnode/.style={fill=darkgray, inner sep = 0pt, minimum size = 5pt},
  every fit/.style={ellipse,draw,inner sep=-2pt,text width=2cm},
  shorten >= 3pt,shorten <= 3pt
]


\node[fsnode] (t1) at (0, 1) [label=left: {\small$\tau_1$}] {};
\node[fsnode] (t2) at (0, 0) [label=left: {\small$\tau_2$}] {};
\node[draw=none, inner sep=0pt] at (0, -0.7) {\vdots};
\node[fsnode] (tn) at (0, -2) [label=left: {\small$\tau_T$}] {};

\node[fsnode] (a1) at (0, -3) [label=left: {\small$a_{11}$}] {};
\node[draw=none, inner sep=0pt] at (0, -3.3) {\vdots};
\node[fsnode] (a2) at (0, -4) [label=left: {\small$a_{bk}$}] {};
\node[draw=none, inner sep=0pt] at (0, -4.3) {\vdots};
\node[fsnode] (a3) at (0, -5) [label=left: {\small$a_{BM_B}$}] {};

\node[ssnode] (c1) at (3.5, 1.5) [label=above: {\small$\gamma_1$}] {};
\node[ssnode] (c2) at (3.5, 0) [label=above: {\small$\gamma_2$}] {};
\node[draw=none, inner sep=0pt] at (3.5, -1) {\vdots};
\node[ssnode] (cn) at (3.5, -2.5) [label=above: {\small$\gamma_C$}] {};

\node [fill = black, inner sep = 0pt, minimum size = 5pt, label=right: \textbf{Overflow}] at (3.5, -5) (overflow) {};

\node [fill = black, inner sep = 0pt, minimum size = 5pt, label=left: \textbf{Source}] at (-3, -1.2) (source) {};

\node [fill = black, inner sep = 0pt, minimum size = 5pt, label=right: \textbf{Sink}] at (5.5, -1.2) (sink) {};


\draw [->] (t1) -- (c1);
\draw [->] (t1) -- (c2);
\draw [->] (t1) -- (cn);

\draw [->] (t2) -- (c1);
\draw [->] (t2) -- (c2);
\draw [->] (t2) -- (cn);

\draw [->] (tn) -- (c1);
\draw [->] (tn) -- (c2);
\draw [->] (tn) -- (cn);

\draw  [->]  (t1) -- (overflow);
\draw  [->]  (t2) -- (overflow);
\draw  [->]  (tn) -- (overflow);

\draw [->] (c1) -- (sink);
\draw [->] (c2) -- (sink);
\draw [->] (cn) -- (sink);

\draw [->, color=blue] (a1) -- (c1);
\draw [->, color=blue] (a1) -- (c2);
\draw [->, color=blue] (a1) -- (cn);

\draw [->, color=blue] (a2) -- (c1);
\draw [->, color=blue] (a2) -- (c2);
\draw [->, color=blue] (a2) -- (cn);

\draw [->, color=blue] (a3) -- (c1);
\draw [->, color=blue] (a3) -- (c2);
\draw [->, color=blue] (a3) -- (cn);

\draw [->] (source) -- (t1);
\draw [->] (source) -- (t2);
\draw [->] (source) -- (tn);
\draw [->] (source) -- (a1);
\draw [->] (source) -- (a2);
\draw [->] (source) -- (a3);
\end{tikzpicture}
\caption{The proposed network structure. Nodes of the form $\tau_t,~t = 1, \dots, T,$ and $\gamma_c,~c = 1, \dots, C,$ correspond to $T$ treated units and $C$ control units, respectively. Nodes of the form $a_{bk},~ b = 1, \dots, B,~k = 1, \dots, M_b,$ are auxiliary nodes that force the fine balance constraint. }
\label{fig: bipartite plot}
\end{figure}

We formulate our proposed one-shot variable-ratio matching with fine balance as a network flow optimization problem. Figure \ref{fig: bipartite plot} displays our proposed network structure. The $T$ treated units are represented by nodes $\{\tau_1, \dots, \tau_T\}$, and the $C$ control units to be matched are represented by nodes $\{\gamma_1, \dots, \gamma_C\}$. For each level $b$ of the nominal variable $X_{\text{nom}}$, we introduce $M_b$ auxiliary nodes. In Figure \ref{fig: bipartite plot}, the collection of nodes $\{a_{bk},b = 1, \dots, B, k = 1, \dots, M_b\}$ denote these auxiliary nodes, totaling $\sum_{b = 1}^B M_b$. Finally, we introduce a \textbf{Source} node, a \textbf{Sink} node and an \textbf{Overflow} node. Write $\mathcal{V}$ for the set of nodes in the proposed network:
\[
\mathcal{V} = \{\tau_1, \dots, \tau_T, \gamma_1, \dots, \gamma_C, a_{11}, \dots, a_{1M_1}, \dots, a_{B1}, \dots, a_{BM_B}, \textbf{Source}, \textbf{Sink}, \textbf{Overflow} \}.
\]
Together, our proposed network contains $|\mathcal{V}| = T + C + 2 + \sum^B_{b = 1} M_b$ nodes.

In a graph, an edge is an ordered pair of nodes, and the network in Figure \ref{fig: bipartite plot} contains six types of edges: (1) edges from the \textbf{Source} node to each treated unit node $\tau_t$, (2) edges from the \textbf{Source} node to each auxiliary node $a_{bk}$, (3) edges from each treated unit node $\tau_t$ to each control unit node $\gamma_c$, (4) edges from each auxiliary node $a_{bk}$ to each control unit node $\gamma_c$, (5) edges from each treated unit node $\tau_t$ to the \textbf{Overflow} node; and (6) edges from each control unit node $\gamma_c$ to the \textbf{Sink} node. Write $\mathcal{E}$ for the following set of edges:
\[
\begin{aligned}
    \mathcal{E} = \bigg\{&(\textbf{Source}, \tau_t), (\textbf{Source}, a_{bk}),
    (\tau_t, \gamma_c), (a_{bk}, \gamma_c), (\tau_t, \textbf{Overflow}), (\gamma_c, \textbf{Sink}): \\
    &t = 1,\dots,T,\; c = 1,\dots,C, \;  b = 1,\dots,B,\; k = 1,\dots,M_b \bigg\} 
\end{aligned}
\]

Each edge $e \in \mathcal{E}$ is associated with a capacity, denoted as $\text{cap}(e)$, which is the maximum units of flow allowed on that edge. Let $L \geq 1$ and $U \geq L$ denote the user-specified, maximum and minimum number of controls, respectively, that can be matched to each treated unit in a matched set. Our designed network in Figure \ref{fig: bipartite plot} has the following specifications for edge capacity:

\begin{enumerate}
    \item $\text{cap}(e) = U$ for $e \in \{(\textbf{Source}, \tau_t):~t = 1, \dots, T\}$;
   
    \item $\text{cap}(e) = 1$ for $e \in \{(\textbf{Source}, a_{bk}), (\tau_t, \gamma_c), (a_{bk}, \gamma_c),(\gamma_c, \textbf{Sink}): t = 1, \dots, T,~ c = 1, \dots, C,~ b = 1, \dots, B,~ k = 1, \dots, M_b\}$;
    
    \item $\text{cap}(e) = U - L$ for $e \in \{(\tau_t, \textbf{Overflow}):~t = 1, \dots, T\}$.
\end{enumerate}

A feasible flow is a map from the set of edges to the integer set $\{0, 1, 2, \dots\}$, denoted as $l : \mathcal{E} \mapsto \{0, 1, 2, \dots\}$, which satisfies the following constraints:

\begin{enumerate}
    \item Edge capacities are respected, in the sense that $0 \leq l(e) \leq \text{cap}(e),~\forall e \in \mathcal{E}$;
    \item The \textbf{Source} node supplies $UT + \sum_{b=1}^{B} M_b$ units of flow:
    \[
    \sum_{t = 1}^T l(\textbf{Source}, \tau_t) + \sum_{b = 1}^B \sum_{k = 1}^{M_b} l(\textbf{Source},a_{bk}) = UT + \sum_{b=1}^B M_b;
    \]
    \item The \textbf{Overflow} node absorbs $UT + \sum_{b=1}^B M_b - C$ units of flow; the \textbf{Sink} node absorbs $C$ units of flow:
\[
\sum_{t=1}^T l(\tau_t, \textbf{Overflow}) = UT + \sum_{b=1}^B M_b - C \quad \text{and} \quad
\sum_{c=1}^C l(\gamma_c, \textbf{Sink}) = C;
\]
    
    \item Flow is conserved for all nodes in $\mathcal{V} \ \backslash \{\textbf{Source}, \textbf{Sink}, \textbf{Overflow}\}$. That is, for any $v \in \mathcal{V} \ \backslash \{\textbf{Source}, \textbf{Sink}, \textbf{Overflow}\}$, the inflow equals the outflow:
\[\sum_{(v^{\prime},v) \in \mathcal{E}} l(v^{\prime}, v) = \sum_{(v,v^{\prime\prime}) \in \mathcal{E}} l(v, v^{\prime\prime}).\] 
\end{enumerate}

Finally, each edge $e \in \mathcal{E}$ is associated with a nonnegative cost denoted as $\text{cost}(e) \geq 0$. In the proposed network structure, we let:
\begin{enumerate}
    \item $\text{cost}(\tau_t, \gamma_c) = \delta_{tc}$ for all $t = 1 \dots T$ and $c = 1, \dots, C$, where $\delta_{tc}$ denotes some measure of covariate distance like Mahalanobis distance or robust Mahalanobis distance between $\tau_t$ and $\gamma_c$ \citep{rosenbaum2010design};
    \item $\text{cost}(\gamma_c, \textbf{Sink}) = \text{cost}(\tau_t, \textbf{Overflow}) = \text{cost}(\textbf{Source}, \tau_t) = \text{cost}(\textbf{Source}, a_{bk}) = 0$ for all $t= 1, \dots, T$, $c = 1, \dots, C$, $b = 1, \dots, B$, and $k = 1, \dots, M_b$;
    \item $\text{cost}(a_{bk}, \gamma_c) = 0$ if the nominal covariate $X_{\text{nom}}$ of $\gamma_c$ is at level $b$, and $\text{cost}(a_{bk}, \gamma_c) =  \infty$ otherwise, for all $c = 1, \dots, C$, $b = 1, \dots, B$, and $k = 1, \dots, M_b$.
\end{enumerate}

For any feasible flow $l$, its cost 
equals $\text{cost}(l) = \sum_{e \in \mathcal{E}} \text{cost}(e)\cdot l(e)$. A feasible flow $l^\star$ is a minimum cost flow if $\text{cost}(l^\star) \leq \text{cost}(l)$ for any feasible flow $l$.

\subsection{Network flow problem; solution}
\label{subsec: methods solution}
A minimum-cost flow can be found in $O(|\mathcal{V}|\cdot|\mathcal{E}| + |\mathcal{V}|^2\log(|\mathcal{V}|))$ operations \citep{korte2011combinatorial}, and the computation can be done via standard network flow algorithms available in open source softwares like \textsf{R} pr \textsf{Python}. In our proposed network, $|\mathcal{V}| = O(C)$ because $C \geq T$ and $\sum_{b=1}^B M_b < C$. Additionally, $|\mathcal{E}| = 2T + C + \sum_{b=1}^B M_b + TC + \sum_{b=1}^B M_b C = O(C^2)$. Hence, the minimum cost can be found in $O(C^3)$ operations. 


Let $\mathcal{F}$ denote the collection of variable-ratio matches that satisfy: (i) each treated unit is matched to between $L$ and $U$ controls; and (ii) the nominal covariate $X_{\text{nom}}$ is finely balanced. The desired optimal variable-ratio match $f^\ast(\cdot) \in \mathcal{F}$ is the one that minimizes the total within-matched-set, treated-to-control covariate distance among all elements in $\mathcal{F}$.

Proposition \ref{prop: f corresponds to l} proves that for any variable-ratio match $f(\cdot) \in \mathcal{F}$, there exists a map $l(\cdot)$ in our proposed network that corresponds to $f(\cdot)$. Therefore, if we let $\mathcal{L}$ denote the collection of all feasible integral flows, then there exists some $l(\cdot) \in \mathcal{L}$ that corresponds to the desired optimal variable-ratio match $f^\ast(\cdot)$ had such an optimal match existed.

\begin{Proposition}\label{prop: f corresponds to l}
    If there exists a variable-ratio match where each treated unit is matched to between $L \geq 1$ and $U \geq L$ control units and the fine balance constraint on $X_{\text{nom}}$ is satisfied, then there exists an integral flow $l(\cdot)$ in the proposed network in Figure \ref{fig: bipartite plot} such that the subclassification induced by $l(\cdot)$ corresponds exactly to this match.
\end{Proposition}
\begin{proof}
Fix a variable-ratio match $f$. We construct an integral flow $l(\cdot)$ corresponding to $f$ and show that $l(\cdot)$ is feasible. Let $l(\tau_t, \gamma_c) = 1$ if the match $f$ matches $\tau_t$ to $\gamma_c$ and $l(\tau_t, \gamma_c) = 0$ otherwise. Then we have $0 \leq l(\tau_t, \gamma_c) \leq 1$, and the capacity constraints on edges of the form $(\tau_t, \gamma_c)$ are satisfied. Because $f(\cdot)$ satisfies the fine balance constraint, $M_b$ control units are removed (not included in the final matched sample) for each level $b \in \{1, \dots, B\}$. Relabel these unmatched controls as $\{\gamma_{b1},\dots, \gamma_{bM_b}\}$, where $\gamma_{bk}$ denotes the $k$-th unmatched control at level $b$. For each unmatched control $\gamma_{bk}$, identify the corresponding auxiliary node $a_{bk}$ and set $l(\textbf{Source}, a_{bk}) = l(a_{bk}, \gamma_{bk}) = 1$ and $l(a_{bk}, \gamma_{b'k'}) = 0$ if $b \neq b'$ or $k \neq k'$. In this way, each auxiliary node $a_{bk}$ transports one unit of flow from the \textbf{Source} node to one control node and the capacity constraints on edges of the form $(a_{bk}, \gamma_c)$ are satisfied. If we further let $l(\textbf{Source}, \tau_t) = U$ for all $\tau_t$, then we have
 \[
    \sum_{t = 1}^T l(\textbf{Source}, \tau_t) + \sum_{b = 1}^B \sum_{k = 1}^{M_b} l(\textbf{Source},a_{bk}) = UT + \sum_{b=1}^B M_b,
\]
satisfying the supply constraint at \textbf{Source}.
For edges from a treated unit to the \textbf{Overflow} node, set $l(\tau_t, \textbf{Overflow}) = U - \#\tau_t$, where $\# \tau_t$ denotes the number of controls matched to $\tau_t$ in $f(\cdot)$. Because $\#\tau_t \in [L, U]$, capacity constraints on edges of the form  $(\tau_t,\textbf{Overflow})$ hold.


Because each control node is matched to one treated or auxiliary node, for each $c = 1, \dots, C$, we have
\[
l(\gamma_c, \textbf{Sink}) = \sum_t l(\tau_t, \gamma_c) + \sum_{b,k} l(a_{bk}, \gamma_c) =  1 \text{ and } \sum_{c=1}^C l(\gamma_c, \textbf{Sink}) = C,
\]
satisfying the demand constraint at \textbf{Sink}.
We have thus shown that the constructed integral flow $l(\cdot)$ satisfies the three constraints defining a feasible flow. 
\end{proof}

\begin{Proposition}
   If $l(\cdot)$ is an integral flow with $\text{cost}(l) < \infty$, then the subclassification induced by $l$, denoted as $ \cup_{t=1}^T\{(\tau_t, \gamma_c): l(\tau_t, \gamma_c) = 1 
\text{ for any } c\}$, is a variable-ratio match that pairs each treated unit with between $L \geq 1$ and $U \geq L$ control units and achieves fine balance on $X_{\text{nom}}$.
\end{Proposition}
\begin{proof}
   Because the \textbf{Source} node emits $UT + \sum_{b=1}^B M_b$ units of flow, and flow is conserved at each treated and auxiliary node, the capacity constraints ensure that each treated node $\tau_t$ transports exactly $U$ units, while each auxiliary node $a_{bk}$, transports exactly one unit.
   As one unit of flow passes through each auxiliary node, there exists a unique $\gamma_c$ such that $l(a_{bk}, \gamma_c) = 1$ because flow is conserved at each node $a_{bk}$. We then verify that $\gamma_c$ must be at level $b$ of $X_{\text{nom}}$. For this $\gamma_c$, because
   \[
   \sum_t l(\tau_t, \gamma_c) + \sum_{b',k'} l(a_{b'k'}, \gamma_c) = l(\gamma_c, \textbf{Sink}) \in \{0,1\},
   \]
we have $l(a_{bk}, \gamma_c) = l(\gamma_c, \textbf{Sink}) = 1$ and $l(\tau_t, \gamma_c) = l(a_{b'k'},\gamma_c) = 0$ for all $t$ and all $(b',k') \neq (b,k)$. If $\gamma_c$ is not at level $b$ of $X_{\text{nom}}$, then $\text{cost}(a_{bk}, \gamma_c) = \infty$, which would imply $\text{cost}(l) = \infty$, contradicting the assumption that $\text{cost}(l) < \infty$. Therefore, $\gamma_c$ must be at level $b$ of the nominal covariate. In total, $\sum_b M_b$ distinct control nodes are matched to auxiliary nodes, with $M_b$ distinct control nodes from each level $b$ of the nominal variable.
By equation \eqref{eq: Mb}, for each level $b$, the number of controls at level $b$ matched to treated nodes is
\[
N_b - M_b = N_b - (N_b - \lfloor \kappa n_b \rfloor) = \lfloor \kappa n_b \rfloor.
\]
Since $\lfloor \kappa n_b \rfloor \approx \kappa n_b$ and the treated group has $n_b$ units at level $b$, the ratio of matched controls to treated units at each level is approximately $\kappa$, achieving fine balance up to rounding error.

The \textbf{Sink} node absorbs \( C \) units of flow, and each unit comes from a distinct control node because $0 \leq l(\gamma_c, \textbf{Sink}) \leq 1$ for $c = 1, \dots, C$. For control nodes not matched to auxiliary nodes, each receives exactly one unit of flow from some treated unit. To see this, let $\gamma_c$ be such a control node. Then $\sum_{b,k} l(a_{bk}, \gamma_c) = 0$, and by flow conservation at $\gamma_c$, \[
\sum_{t} l(\tau_t, \gamma_c) = l(\gamma_c, \textbf{Sink}) = 1.
\]
Since $l(\tau_t, \gamma_c) \in \{0,1\}$ for all $t$, there exists a unique $t^*$ such that $l(\tau_{t^*}, \gamma_c) = 1$. Each treated node $\tau_t$ transports \( U \) units of flow. By flow conservation at $\tau_t$, it follows that
\[
U = l(\textbf{Source}, \tau_t) = \sum_c l(\tau_t, \gamma_c) + l(\tau_t, \textbf{Overflow}).
\]
Since $\text{cap}(\tau_t, \textbf{Overflow}) =U-L$ and $l(\tau_t, \gamma_c) \in \{0,1\}$ for all $c$, we have
\[
0 \leq l(\tau_t, \textbf{Overflow}) \leq U - L \quad \Rightarrow \quad L \leq \sum_c l(\tau_t, \gamma_c) \leq U.
\]
Therefore, each treated node necessarily sends one unit of flow to between \( L \) and \( U \) control units. Hence  $ \cup_{t=1}^T\{(\tau_t, \gamma_c): l(\tau_t, \gamma_c) = 1 
\text{ for any } c\}$ forms the desired variable-ratio match.
\end{proof}

Corollary 1 is an immediate consequence of Proposition 2.

\begin{Corollary}
    Let \( l \) be a minimum-cost integral flow for the network. If \( \text{cost}(l) < \infty \), then the subclassification induced by \( l \), equals $f^\ast(\cdot)$. If \( \text{cost}(l) = \infty \), then $f^\ast(\cdot)$ does not exist.
\end{Corollary}

By Proposition 1, we know some feasible flow $l(\cdot) \in \mathcal{L}$ corresponds to $f^\ast(\cdot)$ had $f^\ast(\cdot)$ existed. By Proposition 2, any feasible flow corresponding to the designed network in Figure \ref{fig: bipartite plot} with a finite cost induces some $f(\cdot) \in \mathcal{F}$. Therefore, to find $f^\ast(\cdot)$, it suffices to find the feasible flow $l(\cdot)$ that minimizes the total cost among all flows in $\mathcal{L}$, as formalized by Corollary 1.

\section{Simulation study}
\label{sec: simulation}
\subsection{Goal, structure, and metrics of success}
\label{subsec: simulation setup}
Our primary goal of the simulation studies is to compare our proposed variable-ratio match method to the ``divide-and-conquer" approach in \citet{pimentel2015variable}. We considered datasets of sample size $n = 3000$. For each unit, we generated a binary treatment assignment $Z$ independently from a Bernoulli distribution with $p = \mathbb{P}(Z=1) = 0.3$ or $0.35$. For each treated unit, we then generated $6$ covariates, $C_1$ to $C_6$, where $C_1$ is drawn from a normal distribution with mean $\mu = 0.25$ or $0.20$ and standard deviation $1$, $C_i$ follows a standard normal distribution for $i = 2, 3, 4, 5$, and $C_6$ is a categorical covariate with $\mathbb{P}(C_6=1) = 0.07$, $\mathbb{P}(C_6=2) = 0.48$ and $\mathbb{P}(C_6=3) = 0.45$. For each control unit, $C_i$ follows a standard normal distribution for $i = 1, 2, 3, 4, 5$, and $C_6$ is a categorical covariate with $\mathbb{P}(C_6=1) = 0.1$, $\mathbb{P}(C_6=2) = 0.5$ and $\mathbb{P}(C_6=3) = 0.4$. Among $5$ continuous covariates, treated and control groups only differed in $C_1$; however, the matching method does not know this, and it matches on all $5$ covariates.

For each generated dataset, we considered the following matching methods:

\begin{enumerate}
    \item $\mathcal{M}^{\text{TS}}$: two-step variable-ratio matching method as described in \citet{pimentel2015variable}. Specifically, we estimated the entire number for each unit by fitting a logistic-regression-based propensity score model, stratified all units based on the estimated entire numbers, and performed a $1$-to-$k$ match with fine balance on $C_6$ for units with entire numbers between $k$ and $k+1$. When it is infeasible to conduct a $1$-to-$k$ match with fine balance, we then attempted a $1$-to-$(k-1)$ match with fine balance, a $1$-to-$(k-2)$ match with fine balance, etc. We let the maximum $k$ be $4$, so that the last stratum consisted of units with estimated entire number $\in [4, \infty)$. When the entire-number-based stratum contained too few units, defined as fewer than $25$, it was then merged to the previous stratum (e.g., stratum $[k, k+1)$ merged to $[k-1, k)$). When conducting each optimal match with fine balance in each stratum, we used the Mahalanobis distance calculated based on covariates $C_1$ to $C_5$.
    

    \item $\mathcal{M}^{\text{OS}}$: one-shot variable-ratio matching method proposed in this article. We considered three different choices of $\kappa$: $\kappa = \kappa_{\max}$, $0.9\kappa_{\max}$, and $0.8\kappa_{\max}$. We will denote these three implementations as $\mathcal{M}^{\text{OS}}_{\kappa_{\max}}$, $\mathcal{M}^{\text{OS}}_{0.9}$ and $\mathcal{M}^{\text{OS}}_{0.8}$, respectively. All three implementations finely balanced $C_6$, although they differed in the size of the final matched comparison group. We set $L = 1$ and $U = 4$, allowing each treated unit to be matched with between $1$ and $4$ control units. All three implementations of $\mathcal{M}^{\text{OS}}$ used the same Mahalanobis distance, calculated based on $C_1$ through $C_5$, as in $\mathcal{M}^{\text{TS}}$. 
\end{enumerate}


The quality of the matched control group produced by each matching algorithm was evaluated using three criteria. First, we computed the standardized mean difference (SMD) of covariate $C_1$, denoted $\textsf{SMD}_{C_1}$. This is defined as the difference in means between the treated and matched control groups, divided by the pooled standard deviation before matching. Second, we measured the total variation distance between the marginal distributions of $C_6$ in the treated and matched control groups, defined as  $\textsf{TV}_{C_6} = \frac{1}{2} \sum_{i \in \{1,2,3\}} |P(i) - Q(i)|$ where $P$ and $Q$ denote the probability mass function in the treated and matched control groups, respectively. When fine balance on $C_6$ is achieved, $\textsf{TV}_{C_6} = 0$. Because covariates $C_2$ through $C_5$ are not associated with treatment assignment, they are expected to be balanced automatically after matching. Third, we recorded the number of control units retained in the final matched comparison group, denoted $n_c$. In general, better matching quality is indicated by a smaller $\textsf{SMD}_{C_1}$, a smaller $\textsf{TV}_{C_6}$, and a larger $n_c$. The simulation study was repeated 200 times.

In addition to these $3$ metrics of success, we also recorded the number of times a $1$-to-$k$ match with fine balance was infeasible in the stratum with entire number between $k$ and $k+1$ (and hence downgraded to a $1$-to-$(k-1)$ match) when implementing $\mathcal{M}^{\text{TS}}$. We also recorded the number of matched pairs, triplets, quadruplets, and quintuplets in the final matched comparison group. Finally, we recorded the computation time for each matching algorithm on a standard laptop computer. 

\subsection{Results}
\label{subsec: simulation results}
Table \ref{tab:OSvDC_2mu} presents the mean and standard deviation of $\textsf{SMD}_{C_1}$, $\textsf{TV}_{C_6}$, and $n_c$ across $200$ simulated datasets under the setting where $p = 0.3$ and $\mu = 0.25$ or $0.20$. Results are shown for each of the three implementations of the proposed one-shot method ($\mathcal{M}^{\text{OS}}$ with $\kappa = \kappa_{\max}$, $0.9\kappa_{\max}$ and $0.8\kappa_{\max}$) and $\mathcal{M}^{\text{TS}}$. The table also summarizes the structure of the matched sets, including how many matched sets are pairs, triplets, quadruplets, and quintuplets. Average computation cost (in seconds) is also reported. Simulation results corresponding to $p = 0.35$ are qualitatively similar to those with $p = 0.30$ and can be found in the Supplemental Material A.

\begin{table}[htb]
\centering
\begin{tabular}{lcccc}
\toprule
& \multicolumn{4}{c}{$p = 0.3,~\mu=0.25$} \\
\cmidrule(lr){2-5}
Metric & $\mathcal{M}^{\text{OS}}_{\kappa_{\max}}$ 
& $\mathcal{M}^{\text{OS}}_{0.9}$ & $\mathcal{M}^{\text{OS}}_{0.8}$ 
& $\mathcal{M}^{\text{TS}}$ \\
\midrule
$\textsf{SMD}_{C_1}$ & 0.22 (0.04) & 0.18 (0.03) & 0.15 (0.03) & 0.26 (0.06) \\
$\textsf{TV}_{C_6}$ & 0.08 (0.06) & 0.05 (0.06) & 0.06 (0.06) & 0.73 (0.45) \\
$n_c$ & 1870 (81) & 1680 (73) & 1490 (65) & 1480 (176) \\
Matched set structure & & & & \\
\quad Pair & 420 (47) & 490 (48) & 570 (49) & 370 (188) \\
\quad Triplet & 170 (14) & 160 (13) & 150 (13) & 480 (197) \\
\quad Quadruplet & 120 (60) & 110 (12) & 90 (12) & 50 (60) \\
\quad Quintuplet & 180 (30) & 140 (25) & 90 (19) & 3 (7) \\
Time (s) & 10.9 (3.3) & 14.2 (1.8) & 21.6 (6.5) & 4.8 (1.6) \\
\midrule
& \multicolumn{4}{c}{$p = 0.3,~\mu=0.20$} \\
\cmidrule(lr){2-5}
Metric & $\mathcal{M}^{\text{OS}}_{\kappa_{\max}}$ 
& $\mathcal{M}^{\text{OS}}_{0.9}$ & $\mathcal{M}^{\text{OS}}_{0.8}$ 
& $\mathcal{M}^{\text{TS}}$ \\
\midrule
$\textsf{SMD}_{C_1}$ & 0.17 (0.04) & 0.15 (0.03) & 0.12 (0.03) & 0.22 (0.07) \\
$\textsf{TV}_{C_6}$ & 0.08 (0.06) & 0.05 (0.06) & 0.06 (0.06) & 0.63 (0.48) \\
$n_c$ & 1867 (81) & 1681 (73) & 1492 (65) & 1491 (207) \\
Matched set structure & & & & \\
\quad Pair & 416 (48) & 491 (48) & 571 (49) & 341 (224) \\
\quad Triplet & 176 (15) & 164 (14) & 148 (14) & 521 (233) \\
\quad Quadruplet & 125 (11) & 108 (12) & 87 (12) & 91 (19) \\
\quad Quintuplet & 181 (29) & 134 (24) & 34 (52) & 1 (3) \\
Time (s) & 10.5 (1.8) & 14.6 (1.8) & 20.9 (2.3) & 7.0 (3.2) \\
\bottomrule
\end{tabular}
\caption{Simulation results when $p = 0.3$ and $\mu = 0.25$ (top panel) or $\mu = 0.20$ (bottom panel). $\mathcal{M}^{\text{OS}}_{\kappa_{\max}}$, $\mathcal{M}^{\text{OS}}_{0.9}$ and $\mathcal{M}^{\text{OS}}_{0.8}$ are three implementations of the proposed one-shot method. $\mathcal{M}^{\text{TS}}$ is the two-step method in \citet{pimentel2015variable}. For each measure of success, mean and standard deviation (in parenthesis) are reported across $200$ simulated datasets.}
\label{tab:OSvDC_2mu}
\end{table}


We highlight two key observations. First, three implementations of the one-shot method--- $\mathcal{M}^{\text{OS}}_{\kappa_{\max}}$, $\mathcal{M}^{\text{OS}}_{0.9}$ and $\mathcal{M}^{\text{OS}}_{0.8}$---consistently outperformed $\mathcal{M}^{\text{TS}}$ across all three key evaluation criteria. In other words, the three one-shot matches uniformly dominated the two-step, divide-and-conquer approach. For example, when $\mu = 0.25$, $\mathcal{M}^{\text{OS}}_{\kappa_{\max}}$ retained $26.4\%$ more controls ($n_c = 1870$ vs $1480$) compared to $\mathcal{M}^{\text{TS}}$ in the final matched sample, while also achieving better covariate balance on $C_1$ and $C_6$. Notably, $C_6$ was finely balanced, up to rounding error, in $\mathcal{M}^{\text{OS}}_{\kappa_{\max}}$, whereas $\mathcal{M}^{\text{TS}}$ failed to achieve fine balance in the final matched sample, despite achieving it within each entire-number-defined stratum.  Second, across the three implementations of the one-shot method, a clear trade-off emerged between the balance on $C_1$ and the size of the matched comparison group: choosing a smaller $\kappa$ improved $\textsf{SMD}_{C_1}$, but at the expense of a reduced $n_c$.

A closer examination of $\mathcal{M}^{\text{TS}}$ reveals that in roughly 40\% times, a $1$-to-$k$ match with fine balance could not be obtained within the corresponding entire-number-defined stratum. As a result, the final matched sets formed by $\mathcal{M}^{\text{TS}}$ consisted primarily of pairs and triplets, whereas the one-shot approach more often yielded quadruplets and quintuplets, contributing to a larger $n_c$.

In terms of computation time, $\mathcal{M}^{\text{TS}}$ is faster: when $p = 0.3$ and $0.25$, despite performing multiple $1$-to-$k$ matches, it took less than $5$ seconds on average, compared to $11$ to $22$ seconds (depending on the choice of $\kappa$) for the one-shot method. Two remarks are in order. First, this difference is not order-wise: both algorithms have the same computational complexity, $O(C^3)$. Second, within $\mathcal{M}^{\text{OS}}$, smaller values of $\kappa$ introduce more auxiliary nodes and edges into the proposed network structure, which in turn increases runtime. Even so, runtimes of $10$ to $20$ seconds for datasets with roughly $3000$ units appear reasonable for most practical applications.


\section{Case study: comparing two matched designs in a study of right heart catheterization}
\label{sec: case study}

We first applied the method of \citeauthor{pimentel2015variable} \citeyearpar{pimentel2015variable} to the RHC dataset. As discussed in Section~\ref{subsec: intro case study and limitation}, we examined the distribution of the insurance type variable after dividing the cohort into four strata based on their estimated entire numbers. For individuals with an entire number in $(0,2)$, we conducted a pair match and obtained $1005$ matched pairs. For those with an entire number in $[2,3)$, a $1$-to-$2$ match with fine balance was infeasible; instead, we performed another pair match, yielding 176 matched pairs. For those in $[3,4)$, we implemented a $1$-to-$3$ match and formed eight matched sets (each with one treated and three controls). Finally, for those with an entire number in $[4,\infty)$, we performed a $1$-to-$2$ match---rather than a $1$-to-$3$ or a $1$-to-$4$ match---since the latter two were both infeasible with fine balance and obtained five matched sets (each with one treated and two controls).

In total, we obtained $1181$ matched pairs, $5$ matched triplets, and $8$ matched quadruples. For this dataset, the divide-and-conquer approach primarily produced matched pairs. Panel B of Table~\ref{tab:realdata_balance} summarizes the covariate balance after matching. Three observations are noteworthy. First, the number of controls only slightly exceeded the number of treated ($1215$ vs.\ $1194$). Second, most standardized mean differences were below $0.1$, corresponding to one-tenth of a pooled standard deviation. Third, although fine balance was achieved within each stratum defined by the entire number, the overall design did not satisfy fine balance when the strata were pooled. Nevertheless, the deviation from fine balance was minimal, largely because the vast majority of matched sets were pairs.

Next, we applied our new algorithm to the same RHC dataset. We first set $\kappa = \kappa_{\max}$ to retain as many controls as possible in the matched comparison group. The final matched design included $1194$ treated patients and $1534$ control patients. Panel A of Table \ref{tab:realdata_balance} summarizes the covariate balance of this design. Relative to the method of \citet{pimentel2015variable}, our approach discarded fewer controls (270 vs.\ 589) and retained substantially more (1534 vs.\ 1215, a $26.3\%$ increase). Moreover, the new design achieved exact fine balance on the insurance type variable. Finally, the standardized mean differences for the remaining covariates were broadly comparable across the two methods. Taken together, the proposed design outperformed the other design, aligning well with the results in the simulation studies.

\begin{table}[!ht]
\centering
\scriptsize
\resizebox{\textwidth}{!}{%
\begin{tabular}{lcccccc}
\toprule
 & \multicolumn{3}{c}{Panel A: $\mathcal{M}^{\text{OS}}_{\kappa_{\max}}$} & \multicolumn{3}{c}{Panel B: $\mathcal{M}^{\text{TS}}$} \\
\cmidrule(lr){2-4} \cmidrule(lr){5-7}
Variable & Control & Treated & SMD & Control & Treated & SMD \\
\midrule
n & 1534 & 1194 &      & 1215 & 1194 &      \\
Age group (\%) & & & 0.049 & & & 0.022 \\
\quad 18-29 & 124 ( 8.1) &  89 ( 7.5) & &  96 ( 7.9) &  89 ( 7.5) & \\
\quad 30-50 & 607 (39.6) & 451 (37.8) & & 465 (38.3) & 451 (37.8) & \\
\quad 51-65 & 803 (52.3) & 654 (54.8) & & 654 (53.8) & 654 (54.8) & \\
Sex = Male (\%) & 877 (57.2) & 693 (58.0) & 0.018 & 694 (57.1) & 693 (58.0) & 0.019 \\
Years of education & 12.17 (2.77) & 12.30 (2.92) & 0.047 & 12.16 (2.66) & 12.30 (2.92) & 0.049 \\
Race (\%) & & & 0.065 & & & 0.054 \\
\quad Black & 267 (17.4) & 219 (18.3) & & 206 (17.0) & 219 (18.3) & \\
\quad Other & 120 ( 7.8) & 112 ( 9.4) & & 102 ( 8.4) & 112 ( 9.4) & \\
\quad White & 1147 (74.8) & 863 (72.3) & & 907 (74.7) & 863 (72.3) & \\
Income (\%) & & & 0.087 & & & 0.060 \\
\quad \$11--\$25k & 269 (17.5) & 236 (19.8) & & 222 (18.3) & 236 (19.8) & \\
\quad \$25--\$50k & 343 (22.4) & 283 (23.7) & & 281 (23.1) & 283 (23.7) & \\
\quad $>$ \$50k   & 162 (10.6) & 134 (11.2) & & 127 (10.5) & 134 (11.2) & \\
\quad Under \$11k & 760 (49.5) & 541 (45.3) & & 585 (48.1) & 541 (45.3) & \\
DASI  & 21.17 (5.76) & 21.13 (5.32) & 0.007 & 21.19 (5.73) & 21.13 (5.32) & 0.010 \\
Cancer (\%) & & & 0.065 & & & 0.012 \\
\quad Metastatic & 107 ( 7.0) &  74 ( 6.2) & &  72 ( 5.9) &  74 ( 6.2) & \\
\quad No         & 1197 (78.0) & 963 (80.7) & & 981 (80.7) & 963 (80.7) & \\
\quad Yes        & 230 (15.0) & 157 (13.1) & & 162 (13.3) & 157 (13.1) & \\
Respiratory rate  & 28.90 (13.91) & 27.76 (14.43) & 0.080 & 28.71 (13.76) & 27.76 (14.43) & 0.068 \\	
PaCo2 & 37.88 (12.14) & 36.74 (11.10) & 0.098 & 37.92 (12.40) & 36.74 (11.10) & 0.101 \\
Temperature & 37.86 (1.72) & 37.78 (1.83) & 0.047 & 37.93 (1.65) & 37.78 (1.83) & 0.048 \\
Urine output & 2283 (1001) & 2315 (1238) & 0.029 & 2273 (993) & 2315 (1238) & 0.038 \\
Urine output miss = 1 (\%) & 843 (55.0) & 598 (50.1) & 0.098 & 645 (53.1) & 598 (50.1) & 0.060 \\
White Blood Cell count & 14.87 (10.58) & 15.77 (12.11) & 0.079 & 14.74 ( 9.56) & 15.77 (12.11) & 0.095 \\
Sodium & 136.41 (7.10) & 136.26 (7.75) & 0.020 & 136.45 (6.69) & 136.26 (7.75) & 0.020 \\
Potassium & 3.99 (1.00) & 4.01 (1.06) & 0.024 & 3.97 (0.96) & 4.01 (1.06) & 0.039 \\
Renal history = 1 (\%) & 76 ( 5.0) & 72 ( 6.0) & 0.047 & 70 ( 5.8) & 72 ( 6.0) & 0.011 \\
Liver history = 1 (\%) & 166 (10.8) & 117 ( 9.8) & 0.034 & 111 ( 9.1) & 117 ( 9.8) & 0.023 \\
Medical Insurance (\%) & & & 0.003 & & & 0.027 \\
\quad Medicaid & 234 (15.3) & 182 (15.2) & & 197 (16.2) & 182 (15.2) & \\
\quad Medicare & 137 ( 8.9) & 107 ( 9.0) & & 107 ( 8.8) & 107 ( 9.0) & \\
\quad Medicare \& Medicaid &  70 ( 4.6) &  55 ( 4.6) & &  55 ( 4.5) &  55 ( 4.6) & \\
\quad No insurance & 145 ( 9.5) & 113 ( 9.5) & & 114 ( 9.4) & 113 ( 9.5) & \\
\quad Private & 869 (56.6) & 675 (56.5) & & 680 (56.0) & 675 (56.5) & \\
\quad Private \& Medicare &  79 ( 5.1) &  62 ( 5.2) & &  62 ( 5.1) &  62 ( 5.2) & \\
\bottomrule
\end{tabular}}
\caption{\textbf{Panel A}: Covariate balance of the variable ratio design using the new method with $\kappa = \kappa_{\max}$. \textbf{Panel B}: Covariate balance using the ``divide-and-conquer" approach in \citet{pimentel2015variable}. SMD: standardized mean difference. Mean and standard deviation (SD) are reported for each continuous variable. Count and percentage (\%) are reported for each categorical variable.}
\label{tab:realdata_balance}
\end{table}

To further assess how the choice of $\kappa$ would influence the resulting matched comparison group, we also applied the proposed method with $\kappa = 0.9\kappa_{\max}$ and $\kappa = 0.8\kappa_{\max}$ (see Supplemental Material B). In each implementation, fine balance on insurance type was achieved. Relative to using $\kappa = \kappa_{\max}$, smaller values of $\kappa$ yielded improved covariate balance and a reduced comparison group size, although the matched comparison group size remained larger than that obtained with \citeauthor{pimentel2015variable}'s \citeyearpar{pimentel2015variable} method.

\section{Summary and discussion}
\label{sec: discussion}
Matching with a variable number of controls is a useful alternative to more conventional $1$-to-$k$ match. It has potential to retain as many control units as possible, while maintaining sufficient covariate balance. In this article, we propose a method that conducts a variable-ratio match with fine balance on one or more important nominal variables in a single shot. The proposed method is shown, via simulation studies and a case study, to outperform the traditional two-step approach in multiple criteria simultaneously, including overall balance, quality of fine balance, and the size of the matched comparison group.  

\citet{chen2023testing} proposed assessing deviations from the randomization assumption by examining covariate balance. Failure to reject the randomization assumption can be taken as evidence that the matched observational study sufficiently emulates a target randomized controlled trial. If the assumption is rejected, researchers may attempt to improve the design---for instance, by selecting a smaller $\kappa$ in our proposed variable-ratio matching method. When no further design improvements are feasible and residual bias persists, outcome analyses that explicitly account for this bias should be used; see, for example, \citet{rosenbaum2002observational}, \citet{fogarty2018regression,fogarty2020studentized} and \citet{chen2023testing} for inferential approaches to matched cohort data under biased randomization schemes.

While the network structure in Figure \ref{fig: bipartite plot} emphasizes key features of variable-ratio matching---such as fine balance and restricting matched set sizes to between $L+1$ and $U+1$---the cost function $\delta(\tau_t, \gamma_c)$ can be modified to accommodate additional design goals. For example, if exact matching on an important prognostic variable $X$ is desired for subsequent subgroup analysis, one may set $\delta(\tau_t, \gamma_c) = \infty$ whenever $\tau_t$ and $\gamma_c$ do not share the same value of $X$. As another example, directional penalties may be applied to $\delta(\tau_t, \gamma_c)$ to help further reduce bias on some recalcitrant covariate \citep{yu2019directional}.

The proposed method can be further accelerated by removing edges from the network. When the sample size is large, rather than connecting each $\tau_t$ to every $\gamma_c$, the network may be sparsified by linking each $\tau_t$ only to a fixed number of candidate controls within the propensity score caliper \citep{yu2020matching}. Since treated and control units with vastly different estimated propensity scores are unlikely to be placed in the same matched set, edges between them can be safely omitted without hurting the quality of the match.

\section{Disclosure statement}\label{disclosure-statement}
The authors have no conflicts of interest.

\section{Data Availability Statement}\label{data-availability-statement}

Deidentified data were publicly available via the following URL: \url{https://hbiostat.org/data/repo/rhc}. The dataset analyzed in the article was obtained from the supplementary material of \cite{rosenbaum2012optimal} and is also available in the supplementary material of the current article.

\phantomsection\label{supplementary-material}
\bigskip

\begin{center}

{\large\bf SUPPLEMENTARY MATERIAL}

\end{center}

\begin{description}
\item[Supplemental Material:]
This supplemental material contains additional simulation results, case study results, and a tutorial on how to install the R package and reproduce results in the paper. (pdf file)
\item[R-package match2C:]
R package match2C is available via Github: \url{https://github.com/bzhangupenn/match2C}.
\item[RHC dataset:]
Dataset used in the case study comes with the R package. To access the dataset, users can use the command \textbf{data(dt\_rhc)} or \textbf{data(dt\_rhc\_ac)}.
\end{description}

\bibliography{bibliography}

\newpage

\end{document}